\documentclass{amc}
\usepackage{amsmath,amsthm,amssymb,amsfonts,amscd,mathrsfs, placeins}
 \usepackage{epsfig} 
 \usepackage{graphics} 

\def\currentvolume{X}
 \def\currentissue{X}
  \def\currentyear{20XX}
   
    \def\ppages{X--XX}
   \def\DOI{10.3934/xx.xx.xx.xx}

\newtheorem{theorem}{Theorem}[section]
\newtheorem{corollary}[theorem]{Corollary}

\newtheorem{lemma}[theorem]{Lemma}
\newtheorem{proposition}[theorem]{Proposition}

\theoremstyle{definition}
\newtheorem{definition}[theorem]{Definition}
\newtheorem{remark}[theorem]{Remark}

\newtheorem{example}[theorem]{Example}

\renewcommand{\bf}{\textbf}
\newcommand{\lp}{\langle}
\newcommand{\rp}{\rangle}

\newcommand{\Z}{\mathbb{Z}}

\newcommand{\cM}{\mathcal{M}}
\newcommand{\cC}{\mathcal{C}}

\newcommand{\hC}{\widehat{C}}

\newcommand{\vb}{\mathbf{b}}
\newcommand{\va}{\mathbf{a}}

\newcommand{\vx}{\mathbf{x}}

\newcommand{\vzero}{\mathbf{0}}

\newcommand{\F}{\mathbb{F}}

\newcommand{\Aut}{Aut}
\newcommand{\Mon}{\textrm{Mon}}

\newcommand{\GL}{\operatorname{GL}}
\newcommand{\SD}{\operatorname{SD}}
\newcommand{\GO}{\operatorname{GO}}
\newcommand{\CO}{\operatorname{CO}}
\renewcommand{\O}{\operatorname{O}}
  \newcommand{\Block}{Block} 
\newcommand{\Mat}{\textrm{Mat}}

\newcommand{\Vector}{\textrm{Vec}}

\newcommand{\Canon}{\textrm{Canon}}

\newcommand{\LStab}{Stab}

\newcommand{\LGOStab}{GOStab}

\newcommand{\tr}{\operatorname{Tr}}

\newcommand{\Matlm}{\F_q^{l \times m}}

\newcommand{\rank}{\operatorname{rank}}
\newcommand{\LEquiv}{Equiv}

\newcommand{\LAut}{Aut}

\title[Enumerating Self-Dual Matrix Codes]{An Enumeration of the Equivalence Classes of Self-Dual Matrix Codes}
\author[Katherine Morrison]{}

\email{ katherine.morrison@unco.edu}%
\thanks{This work was supported in part by NSF grants grants DMS-0735099, DMS-0903517 and DMS-0838463, as well as DOE grant P200A090002.}

\begin{document}
\maketitle

\centerline{\scshape Katherine Morrison }
\medskip
{\footnotesize
    \centerline{School of Mathematical Sciences }
   \centerline{University of Northern Colorado}
   \centerline{501 20th St, CB 122}
   \centerline{ Greeley, CO 80639, USA}
}

\medskip

 \centerline{(Communicated by xxxxxx)}

\begin{abstract}
As a result of their applications in network coding, space-time coding, and coding for criss-cross errors, matrix codes have garnered significant attention; in various contexts, these codes have also been termed rank-metric codes, space-time codes over finite fields, and array codes.  We focus on characterizing matrix codes that are both efficient (have high rate) and effective at error correction (have high minimum rank-distance).  It is well known that the inherent trade-off between dimension and minimum distance for a matrix code is reversed for its dual code; specifically, if a matrix code has high dimension and low minimum distance, then its dual code will have low dimension and high minimum distance.  With an aim towards finding codes with a perfectly balanced trade-off, we study self-dual matrix codes.  In this work, we develop a framework based on double cosets of the matrix-equivalence maps to provide a complete classification of the equivalence classes of self-dual matrix codes, and we employ this method to enumerate the equivalence classes of these codes for small parameters.
\end{abstract}

\section{Introduction}\label{intro}

Codes consisting of matrices over a finite field with the rank distance have been employed in a variety of applications, although often under different names and with a different focus.  Recently, they have garnered attention in the context of error control for network coding due to their role in the construction of \emph{lifted rank-metric codes} \cite{KK08,SKK08}.  Previously, their potential use in space-time coding has been investigated by Grant and Varanasi \cite{GV08, GVSpaceTime}; in this context, these codes are known as \emph{space-time codes over a finite field}.  Finally, these codes, restricted to square matrices, have also been studied by Blaum et al.\ \cite{Blaum} and Roth \cite{Roth} in the context of memory chip arrays and magnetic tape recording, where it is essential to protect against criss-cross errors; in that work, the codes were known as \emph{array codes}.  In each of these cases, the codes of interest were those consisting of matrices over a finite field with the relevant metric being the \emph{rank distance}: the distance between matrices $A$ and $B$ is $d(A, B) := \rank(A - B)$.  To distill the primary coding principles from each of these contexts and make precise the types of codes we investigate here, we introduce the unifying terminology \emph{matrix codes}, for which we will only consider the rank distance.  

Given the many contexts in which matrix codes arise, it is essential that we better understand the structure and distance properties of these codes, with a focus on codes that are both efficient, i.e.\ have high rate, and effective at error correction, i.e.\ have high minimum distance.  The primary focus of this work is to provide a framework for classifying matrix codes based on these defining properties and to perform such a classification for the class of matrix codes that are \emph{self-dual}.   We restrict our focus to self-dual codes because the analogue of the MacWilliams Identities for matrix codes demonstrates that the inherent trade-off between dimension and minimum distance for a code is reversed for its dual code: if a code has high dimension and low minimum distance, then its dual code will have low dimension and high minimum distance.  

In Section \ref{background}, we begin with the notion of duality for matrix codes, highlighting the relationship between this and the notion of duality in the block code case.  We then turn toward classifying self-dual matrix codes in terms of their structural and distance properties.  This leads us to review previous work on the notion of \emph{equivalence} for matrix codes and on the collection of linear equivalence maps for matrix codes.  

Next, we refine the notion of equivalence maps to apply specifically to self-dual matrix codes.  This focus on self-dual matrix codes necessitates an additional property in our definition of equivalence, namely that an equivalence map sends a self-dual code to another self-dual code.  Toward this end, in Section \ref{equivmapsselfdual}, we characterize the subset of matrix-equivalence maps that commute with the dual and thus maintain the property of self-duality. 

In Section \ref{massformulasection}, we give a matrix-code analogue of the mass formula, which enables one to determine when a classification of self-dual codes for a given set of parameters is complete.  We then outline an alternative method for enumerating the equivalence classes of self-dual matrix codes in Section \ref{doublecosetsmethod}, building off an approach proposed in \cite{Janusz07} to enumerate inequivalent self-dual block codes using double-cosets.  Finally, we carry out this enumeration for matrix codes of small lengths over small finite fields and give the results in Section \ref{Enumerateselfdual}.

\section{Background}\label{background}
As described above, matrix codes arise in a number of important contexts indicating the need to characterize collections of matrix codes with good distance properties.   In particular, we seek to classify the collection of self-dual codes since 
the analogue of the MacWilliams Identities \cite{Delsarte, GV08} implies they have significant potential for good distance distributions.  To achieve this goal, we begin with the notion of duality for linear matrix codes, where a matrix code $C \subseteq \Matlm$ is deemed \emph{linear} if the collection of codewords forms a vector space over $\F_q$.

\begin{definition} [\cite{GV08}]\label{DualDef}
Let $C\subseteq \Matlm$ be a linear matrix code.  The \emph{matrix dual code of $C$} is given by 
\[
C^{\perp_{\Mat}} = \{Y \in \Matlm ~|~\tr(XY^\top) = 0 \text{ for all }
X \in C\}.
\]
We say that $C$ is a \emph{self-dual matrix code} if $C= C^{\perp_{\Mat}}$.
\end{definition}

To enable a comparison of the duals of matrix codes with those of block codes, we must introduce a natural map translating between these types of codes.  
\begin{definition}\label{rhomap}
For $X=[x_{ij}]\in \F_q^{l \times m}$, the \emph{extended row vector corresponding to $X$} is the vector 
\[
\rho(X) = (x_{1 1}, \dots, x_{1 m}, \dots, x_{l 1},\dots,
x_{l m}) \in \F_q^{lm}
\]
formed by concatenating the rows of the $l\times m$ matrix $X$.  If $C\subseteq \Matlm$ is a matrix code, the \emph{extended block code of $C$} is given by
\[
\rho(C)= \{\rho(X) \, | \, X \in C\} \subseteq \F_q^{lm}.
\]
\end{definition}

Since 
$\tr(XY^\top) = \rho(X)\cdot \rho(Y)$, where $\cdot$ denotes the standard dot product, the matrix dual code of $C$ is simply the inverse image under $\rho$ of the standard block dual code of $\rho(C)$, i.e.\ 
\[
C^{\perp_{\Mat}} = \rho^{-1}\left(\rho(C)^\perp \right).
\]
Thus, there is a natural correspondence between the duals of matrix codes and the duals of block codes.  The map $\rho$ also enables a notion of a generator matrix for a linear matrix code:
we will say that $G$ is a \emph{generator matrix for a linear matrix code $C$} if $G$ is a generator matrix for the corresponding extended block code $\rho(C)$, i.e.\ for all $X \in C$, $X = \rho^{-1}(\vx)$ for some $\vx \in \text{rowspace}(G)$.  

It is important to note that although the dual codes coincide, the distance metrics for block codes are different from those for matrix codes.  Throughout, we consider only the Hamming distance for block codes and the rank distance for matrix codes.   Given this difference in metrics, the literature regarding the distance distributions of the dual codes of block codes does not apply directly.  In particular, one major result in this area is the MacWilliams Identities for block codes, which give an explicit relationship between the distance distribution of a code and that of its dual code.  Although these results do not apply to matrix codes, Delsarte as well as Grant and Varanasi have proven analogues of these identities for the case of matrix codes \cite{Delsarte, GV08}.  An important consequence of these results is that the inherent trade-off between dimension and minimum distance for a matrix code is reversed for its dual code.  
As we seek matrix codes with a balanced trade-off between dimension and minimum distance, we turn to the collection of self-dual matrix codes.  In particular, we seek to classify the collection of self-dual codes based on their structure and distance properties.  We must first review a notion of equivalence for matrix codes generally, and then we develop a refined notion of equivalence for self-dual matrix codes.  

\subsection{Equivalence of Matrix Codes}\label{equivalence}
Intuitively, two codes should be considered equivalent if they share all the same properties and structure.  In particular, equivalent codes should have the same distance distribution and the same number of codewords, or same dimension if the codes are linear.   In the case of block codes, the notion of code equivalence was made more tractable by restricting the definition to declare two block codes to be \emph{equivalent} if and only if there exists a linear, or more generally a semi-linear, invertible map between them that preserves (Hamming) weight.  For this paper, we will restrict to considering linear equivalence only, but all the following results can be generalized to the case of semi-linear equivalence as well; the interested reader may find these results in \cite{MyThesis}.

It is a consequence of the MacWilliams Extension Theorem \cite{MacWilliamsThesis} that the monomial matrices are the only Hamming-weight preserving linear maps, where a \emph{monomial matrix} is any matrix of the form $DP$ for some invertible diagonal matrix $D$ and permutation matrix $P$.  Thus, the collection of linear equivalence maps for block codes consists of only these matrices.

A notion of equivalence for matrix codes was first given in \cite{MyEquivMapsPaper}.  Two matrix codes are \emph{linearly equivalent} if there exists an invertible linear map between them that preserves rank weight of all matrices in $\Matlm$; we call such a map a \emph{linear matrix-equivalence map}.  Note that this requirement of preserving rank weight for \emph{all} matrices is stronger than simply requiring the map to preserve rank weight for the matrices in a given code.  In the case of Hamming weight, these two notions coincide as a result of the MacWilliams Extension Theorem \cite{MacWilliamsThesis}; however, no analogue of this extension theorem holds for rank weight \cite{Heide}. 

The collection of rank-preserving linear maps on $\Matlm$ consists of compositions of multiplication on the left by matrices in $\GL_l(\F_q)$, multiplication on the right by matrices in $\GL_m(\F_q)$, and, when $l=m$, transposition  \cite{Moyls, MyEquivMapsPaper}.   To further the analogy with linear maps on block codes, we translate the action of maps on $\Matlm$ to maps on $\F_q^{lm}$; this translation will also prove useful for the enumeration of self-dual matrix codes.  For this purpose, we return to the map $\rho$ from Definition \ref{rhomap}, which takes a matrix to the extended row vector formed by concatenating the rows of the matrix.  In this context, we are able to describe the linear matrix-equivalence maps in terms of linear maps that
act on the right only, and thus situate the maps in the same universe.  Left multiplication by $L  \in \GL_l(\F_q)$ on
$A \in \Matlm$ corresponds to right multiplication by $(L^\top \otimes I_m)$ on $\va=\rho(A)$ where $I_m$ is the $m \times m$
identity matrix and $\otimes$ denotes the Kronecker product of
matrices.  
Similarly, right multiplication by $M \in \GL_m(\F_q)$ on $A \in \Matlm$ corresponds to right multiplication by $(I_l \otimes M)$ on $\rho(A)$.
Finally, transposition corresponds to multiplication on the right by the $m^2 \times m^2$ block matrix $T=[E_{ji}]_{ij}$ whose $(i,j)^\text{th}$ block is the $m \times m$ matrix $E_{ji}$, for $1 \leq i,j\leq m$, where $E_{ji}$ is the matrix of all zeros with a single 1 in the $(j,i)^\text{th}$ entry.   

This change of context motivates the following extension of the definition of matrix-equivalence maps.
 \begin{definition}\label{equivmapsonvectors}
We say that  $f:\F_q^{lm} \to \F_q^{lm}$ is a \emph{linear $l \times m$ matrix-equivalence map on extended row vectors} if there is some linear matrix-equivalence map $g:\Matlm \to \Matlm$ such that, for all $A \in \Matlm$, $f(\rho(A))= \rho(g(A))$, i.e.\  $f= \rho \circ g \circ \rho^{-1}$.
We denote the collection of linear $l \times m$ matrix-equivalence maps on extended row vectors by $\LEquiv_{\Vector}(\Matlm)$.
\end{definition}

From the previous commentary and the fact that $\lambda I_l \otimes \lambda^{-1}I_m$ is the identity map on $\F_q^{lm}$ for any $\lambda \in \F_q^*$, it is easy to show that the collection of matrix-equivalence maps on extended row vectors has the group structure given in Proposition \ref{LEquivVec} below, and so we omit the proof here; a full proof is available in \cite{MyThesis} for the interested reader.

 \begin{proposition}\label{LEquivVec}
The group $\LEquiv_{\Vector}(\Matlm)$ of linear $l \times m$ matrix-equivalence maps on extended row vectors in $\F_q^{lm}$ satisfies
\[
\LEquiv_{\Vector}\left(\Matlm\right) =
\left\{ \begin{array}{ll} \{T^i(L^\top \otimes M) ~|~i \in \{0,1\},~L,M \in \GL_l(\F_q)\} & \text{if } l=m \\ 

\{L^\top \otimes M ~|~L\in \GL_l(\F_q),~M \in \GL_m(\F_q)\} & \text{if } l \neq m 
\end{array} \right.
\]
where $T=[E_{ji}]_{ij}$ is the matrix for transposition.  The group structure is given by
\[
\LEquiv_{\Vector}(\Matlm) \cong  \left\{ \begin{array}{ll} \Z_2 \ltimes \left(\GL_l(\F_q) \times \GL_l(\F_q)\right)/N & \text{if } l=m \\ 

\left(\GL_l(\F_q) \times \GL_m(\F_q)\right)/N & \text{if } l \neq m 
\end{array} \right.
\]
where $N=\{(\lambda I_l, \lambda^{-1}I_m)~|~\lambda~\in~\F_q^*\} \leq \GL_l(\F_q) \times \GL_m(\F_q)$.  
  \end{proposition}
  
  \begin{remark}
By Proposition \ref{LEquivVec}, each linear matrix-equivalence map on extended row vectors has the form $T^i(L^\top \otimes M)$ although there are many maps that have the same action as $T^i(L^\top \otimes M)$ since $T^i(L^\top \otimes M)$ produces the same effect as $T^i(L^\top \otimes M)(\lambda I_l \otimes\lambda^{-1}I_m)$ for any $\lambda \in \F_q^*$.  For simplicity, we will use the notation $[T^i(L^\top \otimes M)]$ to denote the equivalence class of maps corresponding to $T^i(L^\top \otimes M)\cdot \lp (\lambda I_l \otimes\lambda^{-1}I_m)~|~\lambda~\in~\F_q^*\rp$.  
\end{remark}

\section{Method for Enumerating Self-Dual Matrix Codes}
The goal of this paper is to enumerate all the linearly-inequivalent $q$-ary self-dual $l\times m$ matrix codes for small values of $q$, $l$, and $m$.  The current definition of code equivalence is too weak for this purpose, however, because it is possible for self-dual codes to be considered equivalent to codes that are not self-dual.  
Section \ref{equivmapsselfdual} below gives a sufficient condition that may be added to the definition of equivalence to ensure that the equivalence classes of self-dual codes will consist solely of self-dual codes.  Also in that section is a characterization of the matrix-equivalence maps that satisfy this additional condition together with a review of the block-equivalence maps that achieve this condition.  
\subsection{Equivalence Maps that Commute with the Dual}\label{equivmapsselfdual}
Recall from Section \ref{equivalence} 
that we have required equivalence maps to be invertible, linear, and weight-preserving.  Motivated by \cite[p. 185]{SelfDualChapter}, we now will also require that an equivalence map $f$ satisfies $f(C^{\perp})=(f(C))^{\perp}$ for every linear code $C$, since this property of commuting with the dual is precisely what is needed to ensure self-dual codes are mapped to other self-dual codes.   

The subset of linear block-equivalence maps that satisfy this property of commuting with the dual, i.e.\ preserving orthogonality, will be known as \emph{linear block-equivalence maps for self-dual codes} and denoted by $LEquiv_{\Block}^{\SD}(\F_q^n)$.
Recall from Section \ref{equivalence} that the set of general linear block-equivalence maps is $\Mon_n(\F_q)$, the group of monomial matrices, i.e.\ the group of invertible matrices of the form $DP$ for some diagonal matrix $D$ and permutation matrix $P$.  We are interested in the following subgroup of $\Mon_n(\F_q)$:

  \begin{definition}\label{scalarmondef} The \emph{scalar monomial group} is the subgroup $\cM_n(\F_q)$ of $\Mon_n(\F_q)$ consisting of matrices of the form $DP$ where $P$ is a permutation matrix and $D$ is a diagonal matrix whose diagonal entries are $\pm \alpha$ for some $\alpha \in \F_q^*$.
  \end{definition}
  
MacWilliams shows \cite{MacWilliamsThesis} that the matrices in $\cM_n(\F_q)$ are the only monomial matrices that commute with the dual, and so $\LEquiv_{\Block}^{\SD}(\F_q)=\cM_n(\F_q)$. 
In the matrix code setting, we have:
\begin{definition}\label{matrixequivalenceSD}
 An invertible map $f:\Matlm \to \Matlm$ is a \emph{linear matrix-equivalence map for self-dual codes} if $f$ is $\F_q$-linear, preserves rank weight, and has the property that for all linear codes $C \subseteq \Matlm$, $f(C^{\perp_{\Mat}}) = (f(C))^{\perp_{\Mat}}$.  The collection of linear matrix-equivalence maps for self-dual codes is denoted $\LEquiv_{\Mat}^{\SD}(\Matlm)$ and we say that two self-dual matrix codes $C,~\hC \subseteq \Matlm$ are \emph{linearly matrix-equivalent} if there exists a linear matrix-equivalence map for self-dual codes $f$ such that $\hC= f(C)$.  Finally, the collection of linear matrix-equivalence maps on extended row vectors for self-dual codes is denoted $\LEquiv_{\Vector}^{\SD}(\Matlm)$ and it consists of all maps $f$ for which there exists $g \in \LEquiv_{\Mat}^{\SD}(\Matlm)$ such that, for all $A \in \Matlm$, $f(\rho(A))= \rho(g(A))$, i.e.\  $f= \rho \circ g \circ \rho^{-1}$.
  \end{definition}
  
We will use the orthogonal similitudes group to characterize the matrix-equivalence maps on extended row vectors that commute with the dual:

\begin{definition}[\cite{Vinroot}]\label{DefOrthog}
The \emph{orthogonal similitudes group}\footnote{There does not seem to be standard notation for this group, and so we have chosen to follow \cite{Taylor, Vinroot}.  The documentation for Magma refers to this group as the \emph{conformal orthogonal group} denoted by $CO_n(\F_q)$.}, also called the \emph{general orthogonal group} \cite[p. 136]{Taylor}, is the collection of matrices $$\GO_n(\F_q)=\{A \in \GL_n(\F_q)~|~AA^\top  = \lambda I_n \text{ for some }\lambda \in \F_q^*\}$$  For $A \in \GO_n(\F_q)$ with $AA^\top  = \lambda I_n$, we call $\lambda$ the \emph{similitude character} of $A$.  The subgroup $\O _n(\F_q)$ of $\GO_n(\F_q)$ consisting of matrices with similitude character 1 is the \emph{orthogonal group}.
\end{definition}
\begin{remark} 
The orthogonal similitudes group may also be defined more generally on any vector space $V$ with a symmetric bilinear form $\lp \cdot , \cdot \rp$ \cite{Grove, Janusz07, Taylor, Vinroot}.  In this case, the orthogonal similitudes group is the collection of matrices that preserve that bilinear form up to a scalar, again 
called the similitude character.  The orthogonal group is the subgroup consisting of matrices with similitude character 1.  Note that the definition of the orthogonal similitudes group given in Definition \ref{DefOrthog} is the special case of this definition for $V= \F_q^n$ with the standard dot product.  In classical group theory, a different symmetric bilinear form (or in the case of characteristic 2, a quadratic form) is often used when defining the orthogonal group \cite[p. 39, 113]{Grove}.  When $n$ is even there are two equivalence classes of symmetric bilinear forms, which give rise to non-isomorphic groups; we show in the appendix which group is conjugate to the group that preserves the standard dot product that we consider here.
\end{remark}

We now turn to a characterization of the matrix-equivalence maps on extended row vectors that commute with the dual. %


  \begin{proposition}\label{LEquivSDVec}
The group $\LEquiv_{\Vector}^{SD}(\Matlm)$ of linear $l \times m$ matrix-equivalence maps for self-dual codes on extended row vectors satisfies
\[
\LEquiv_{\Vector}^{SD}(\Matlm) \cong  \left\{ \begin{array}{ll} \Z_2 \ltimes \left(\GO_l(\F_q) \times \GO_l(\F_q)\right)/N & \text{if } l=m \\ 

\left(\GO_l(\F_q) \times \GO_m(\F_q)\right)/N & \text{if } l \neq m, 
\end{array} \right.
\]
where $N=\{(\lambda I_l, \lambda^{-1}I_m)~|~\lambda~\in~\F_q^*\} \leq \GO_l(\F_q) \times \GO_m(\F_q)$.  
%
 \end{proposition}
 \begin{proof}
 By Proposition \ref{LEquivVec}, we have 
 \[
 \LEquiv_{\Vector}\left(\Matlm\right) =\left\{\begin{array}{ll}  \{T^i(L^\top \otimes M) ~|~i \in \{0,1\},~L,M \in \GL_l(\F_q)\}, & \text{if } l=m\\
 \{L^\top \otimes M ~|~L\in \GL_l(\F_q),~M \in \GL_m(\F_q)\}, & \text{if } l\neq m. \end{array} \right.
 \]
%
Thus, we need only determine which maps in this subgroup commute with the matrix dual.  Note that these maps act on $\rho(C)$ for any $C \subseteq \Matlm$, and so we need to express the matrix dual code $C^{\perp_{\Mat}}$ in terms of $\rho(C)$.  By the commentary following Definition \ref{rhomap},  we see that $\rho(C^{\perp_{\Mat}})=(\rho(C))^\perp$ where $\perp$ denotes the standard block code dual.  Thus, it suffices to determine which maps in $ \LEquiv_{\Vector}\left(\Matlm\right)$ commute with the standard block code dual.  
We will handle $l=m$ and $l \neq m$ simultaneously by writing $ \LEquiv_{\Vector}\left(\Matlm\right)= \{T^i(L^\top \otimes M) ~|~ i\in\{0,1\},~L\in \GL_l(\F_q),~M \in \GL_m(\F_q)\}$ where $i=0$ if $l \neq m$.   

Let $f\in \LEquiv_{\Vector}^{SD}(\Matlm)$.  Then since $\LEquiv_{\Vector}^{SD}(\Matlm) \subseteq \LEquiv_{\Vector}(\Matlm)$, there exists some $L\in \GL_l(\F_q)$ and $M \in \GL_m(\F_q)$ such that $f=T^i(L \otimes M)$, and $f$ satisfies $f((\rho(C))^\perp)=(f(\rho(C)))^\perp$
 for every linear matrix code $C \subseteq \Matlm$. Since $C$ is a linear matrix code, there is a generator matrix $G_C\in \F_q^{k \times lm}$ such that $\rho(C) = \{\va G_C~|~\va \in \F_q^k\}$ where $k=\dim_{\F_q}(C)$.  Then $(\rho(C))^\perp$ has $G_C$ as its parity-check matrix, i.e.\ $(\rho(C))^\perp=\{\vb \in \F_q^{lm}~|~G_C\vb^\top =\vzero\}$, and so 
 \[
f((\rho(C))^\perp) = \{\vb T^i(L^\top \otimes M)~|~G_C\vb^\top =\vzero\}.
\]
Since $G_CT^i(L^\top \otimes M)$ is a generator matrix for $f(\rho(C))$, we have 
 \[
(f(\rho(C)))^\perp = \{\vb ~|~G_CT^i(L^\top \otimes M)\vb^\top =\vzero\}.
\]
Since $f((\rho(C))^\perp)=(f(\rho(C)))^\perp$, each $\vb T^i(L^\top \otimes M) \in f((\rho(C))^\perp)$ with $G_C\vb^\top =\vzero$ must satisfy $G_CT^i(L^\top \otimes M) (\vb T^i(L^\top \otimes M))^\top = \vzero$.  Hence, $G_C$ and $G_CT^i(L^\top \otimes M)(T^{i}(L^\top \otimes M))^\top$ must have the same nullspace for every possible generator matrix $G_C$, and so $T^i(L^\top \otimes M)(T^{i}(L^\top \otimes M))^\top = \lambda I_{lm}$ for some $\lambda \in \F_q^*$ since any column operations other than global multiplication would change the nullspace of at least one generator matrix $G_C$.
Thus,
 \[
\lambda I_{lm}= T^i(L^\top \otimes M)(T^{i}(L^\top \otimes M))^\top = T^i(L^\top \otimes M)(L  \otimes M^\top )T^{i\top}= T^i(L^\top L \otimes MM^\top )T^{i\top}.
 \]
Using the facts that $T$ is symmetric and $T^2=I_{lm}$, we have 
 \[
\left(L^\top L \otimes MM^\top \right)=\lambda (T^{-i})^2=  \lambda I_{lm}.
 \]
Since $\lambda I_{lm} = \lambda (I_l \otimes I_m)$ and decomposition into Kronecker products is unique up to scalars, we have that $L^\top L  = \lambda_1 I_l$ and $MM^\top =\lambda_2I_m$ where $\lambda~=~\lambda_1\lambda_2 ~\neq~0$.  Thus, $L\in \GO_l(\F_q)$ and $M \in \GO_m(\F_q)$, and so 
\[
\LEquiv_{\Vector}^{SD}(\Matlm) = \left\{\begin{array}{ll}  \{T^i(L^\top \otimes M) ~|~i \in \{0,1\},~L,M \in \GO_l(\F_q)\}, & \text{if } l=m\\
 \{L^\top \otimes M ~|~L\in \GO_l(\F_q),~M \in \GO_m(\F_q)\}, & \text{if } l\neq m. \end{array} \right.
\]
Finally, using the group structure of the Kronecker product and the semi-direct product structure for the action of the transpose matrix, we obtain the desired isomorphism.

 \end{proof}
 
  \begin{remark}
From the proof of Proposition \ref{LEquivSDVec}, we see that the linear matrix-equivalence maps for self-dual codes correspond to the maps in $\langle T \rangle \ltimes  \left(\GO_l(\F_q) \otimes \GO_m(\F_q)\right)$, a subgroup of $\GO_{lm}(\F_q)$.  Meanwhile, the linear block-equivalence maps for self-dual codes are the scalar monomial matrices $\cM_{lm}(\F_q)$, also a subgroup of $\GO_{lm}(\F_q)$.  These two groups do not coincide in general.  For example, over odd characteristic, when $l=m=2$, the block matrix 
 $${\footnotesize\left[\begin{array}{rrrr}1&-1&0&0\\1&1&0&0\\0&0&1&-1\\0&0&1&1  \end{array}\right]}= I_l \otimes \left[\begin{array}{rr} 1&-1\\1&1\end{array}\right]$$
  is an element of $\GO_l(\F_q) \otimes \GO_m(\F_q)$ and thus also an element of $\GO_{lm}(\F_q)$; however, this matrix is not an element of $\cM_{lm}(\F_q)$.  A similar trick can be used to show the distinction between these groups for any other value of $l>1$.  Additionally, there are matrices in $\cM_{lm}(\F_q)$ that cannot be written as a Kronecker product, and thus, are not elements of $\GO_l(\F_q) \otimes \GO_m(\F_q)$ or of $\langle T \rangle \ltimes  \left(\GO_l(\F_q) \otimes \GO_m(\F_q)\right)$.  
 \end{remark}

\section{Characterizing Linear Equivalence Classes\\ of Self-Dual Matrix Codes}
We now investigate the linear-equivalence classes of self-dual matrix codes of small parameters.  A similar such analysis of self-dual block codes was an important achievement in classical coding theory.  One practical motivation for this analysis arises from the MacWilliams Identities, which show that the relationship between dimension and minimum distance for a code is reversed for its dual code.  As a result of this reversed relationship, many self-dual codes attain a perfectly balanced trade-off between dimension and minimum distance, making some self-dual codes the best codes known for given parameters \cite{Pless}.  This analysis continued for suboptimal self-dual codes as well, however, because of the nice mathematical properties of self-dual codes, which have produced applications in groups, lattices, and designs \cite{MallowsPlessSloane1976}.  In the sequel, we follow roadmaps from the study of self-dual block codes to enumerate representatives of the equivalence classes of self-dual matrix codes.

Section~\ref{massformulasection} focuses on one key tool that was used to enumerate self-dual block codes, namely the mass formula.  This formula was essential for allowing researchers to determine if all the inequivalent self-dual block codes they had found through various techniques actually gave a complete enumeration of the equivalence classes.  We give an analogous mass formula for self-dual matrix codes that is virtually identical to that from the block code case, taking into account the different equivalence maps for matrix codes.  These formulas align so closely because they are both direct applications of the Orbit-Stabilizer theorem and because there is a one-to-one correspondence between self-dual block codes and self-dual matrix codes as demonstrated in Section~\ref{massformulasection}.  

If one has already generated a number of constructions for self-dual codes and simply needs to confirm that all equivalence classes have been accounted for, then the mass formula is particularly valuable; however, this formula does not actually give a method for producing/finding these self-dual codes.  In Section~\ref{doublecosetsmethod}, we turn to an alternative method for enumerating inequivalent self-dual codes that uses double cosets.  This method will prove more useful here because it will not require the development of multiple new constructions of self-dual codes; we will show that we can construct a single \emph{canonical} self-dual code for each set of parameters and this will suffice to produce a list of all inequivalent self-dual matrix codes.  

\subsection{Mass Formula for Self-Dual Matrix Codes}\label{massformulasection}

As mentioned above, the mass formula for self-dual block codes is used to determine whether a given set of linearly-inequivalent self-dual codes 
is a complete representation of all the equivalence classes of such codes.  This formula comes directly from the Orbit-Stabilizer theorem of abstract algebra, which can be found in any introductory graduate text, e.g.\ \cite[Proposition 2, p.\ 114]{DummitFoote}.  For completeness of the duality theory, we give an analogue of the mass formula for self-dual matrix codes, although we follow a different method, described in Section \ref{doublecosetsmethod}, to actually enumerate the equivalence classes of self-dual matrix codes.  

We begin by observing that self-dual matrix codes are in bijective correspondence with self-dual block codes of appropriate dimensions.

 \begin{lemma}\label{SelfDualBlockMatrix}
 Let $C \subseteq \Matlm$ be a matrix code with corresponding block code $\rho(C) \subseteq \F_q^{lm}$.  Then $C$ is a self-dual matrix code if and only if $\rho(C)$ is a self-dual block code.
 \end{lemma}
 \begin{proof}
This is an immediate consequence of the commentary following Definition \ref{rhomap}, which showed that $\rho(C^{\perp_{\Mat}})=(\rho(C))^\perp$.

 \end{proof}

%

 \begin{corollary}\label{NumSelfDualMat}
 The number of $[l \times m, \frac{lm}{2}]_q$ self-dual matrix codes is$$b \prod_{i=1}^{\frac{lm}{2}-1}(q^i+1)$$ where for even $q$, we set $b=1$, and for odd $q$, we set $b=\left\{\begin{array}{ll} 0 & \textrm{if } q \equiv 3 \pmod 4 \textrm{ and } 4\nmid lm\\ 2& \textrm{otherwise} \end{array}\right.$
 \end{corollary}
 \begin{proof}
 It was shown in \cite[p.\ 184]{SelfDualChapter} that the formula above counts the number of self-dual block codes, and so by Lemma \ref{SelfDualBlockMatrix}, the formula counts the number of self-dual matrix codes as well.
 \end{proof}

Next we observe that the group of equivalence maps on self-dual codes acts on the collection of self-dual codes, and
the \emph{orbits} under this action are precisely the equivalence classes of self-dual codes.  
The \emph{stabilizers} under the action are the \emph{automorphism groups} of the codes restricted to the set of linear matrix-equivalence maps for self-dual codes.  We denote such an automorphism group for a self-dual code $C$ by $\Aut_{\Mat}^{\SD}(C)$.
With this terminology in place, we may now state the mass formula analogue for self-dual matrix codes.

\begin{theorem}[Mass Formula for Matrix Codes]\label{MassMatrix}
 Let $\{C_i\subseteq \Matlm\}_{i=1}^{r}$ be representatives of the distinct equivalence classes of linearly matrix-equivalent self-dual matrix codes of size $l \times m$ over $\F_q$. Then
$$ b \prod_{i=1}^{\frac{lm}{2}-1}(q^i+1)= \sum_{\textrm{linearly matrix-inequivalent } C_i} \frac{|\LEquiv_{\Mat}^{SD}(\Matlm)|}{|\LAut_{\Mat}^{SD}(C_i)|},$$
where for even $q$, we set $b=1$, and for odd $q$, we set $b=\left\{\begin{array}{ll} 0 & \textrm{if } q \equiv 3 \pmod 4 \textrm{ and } 4\nmid lm\\ 2& \textrm{otherwise} \end{array}\right.$
\end{theorem}
\begin{proof}
The space of self-dual matrix codes is partitioned into equivalence classes, or orbits, by the action of the linear matrix equivalence maps.   The number of codes is the sum of the sizes of the orbits.  By the Orbit-Stabilizer theorem, see e.g.\ \cite[Proposition 2, p.\ 114]{DummitFoote}, the size of the orbit of a code $C$ equals $\displaystyle \frac{|\LEquiv_{\Mat}^{SD}(\Matlm)|}{|\LAut_{\Mat}^{SD}(C)|}$.  The result follows. 
\end{proof}

As noted above, the mass formula plays a crucial role in the enumeration of linearly inequivalent self-dual block codes.  
The mass formula is only useful, however, if one has developed a number of different constructions of self-dual codes.  While we could take advantage of known constructions from the block code literature, we take a slightly different approach here that enables us to exploit a single construction for a \emph{canonical} self-dual matrix code for each set of parameters. This alternative approach is described in the next section.

\subsection{Double-Coset Characterization of Inequivalent Self-Dual Matrix Codes}\label{doublecosetsmethod}
In this section, we extend the work of Janusz \cite{Janusz07} to recast the problem of enumerating linearly inequivalent matrix codes as a problem of enumerating double-coset representatives within a group that acts transitively on the collection of self-dual codes.  While the problem of enumerating double-coset representatives is also intractable for large groups, it has proven feasible in the small cases we will examine here whereas a brute force search is not possible for many of these cases.  This framework also provides a nice bridge for relating the equivalence classes of self-dual block codes with those of self-dual matrix codes as well as relating the automorphism groups of the codes in each of these contexts.    Toward this end, we introduce a number of theorems and terminology given in Janusz's paper for the case of self-dual block codes over the field $\F_2$, and we extend these ideas to work for both block and matrix codes over arbitrary finite fields. 

Janusz shows \cite{Janusz07} that $\O_n(\F_2)$ acts transitively on the collection of self-dual block codes over $\F_2$.  In Theorem \ref{GOTransitive}, we prove that $\GO_n(\F_q)$ acts transitively on the collection of self-dual block and matrix codes over arbitrary fields.  We include the proof here because of this modification, but also because the following proof is dramatically simpler than the proof given in \cite{Janusz07}.  

\begin{theorem}[Extended from \cite{Janusz07}]\label{GOTransitive}
The group $\GO_n(\F_q)$ of orthogonal matrices acts transitively on the set of self-dual codes of length $n$ over $\F_q$.  That is:
\begin{enumerate}
\item If $C\subseteq \F_q^n$ is self-dual and $A \in \GO_n(\F_q)$, then $C A$ is self-dual, and
\item If $C,~\widehat{C} \subseteq \F_q^n$ are self-dual, then there is some $A \in \GO_n(\F_q)$ such that\\
 $C A=\widehat{C}$.
\end{enumerate}
\end{theorem}
\begin{proof}
Let $G$ be a generator matrix for the self-dual code $C$.  Then $GG^\top = {\bf 0}_k$, where $k=\frac{n}{2}$ is the dimension of the code.  Let $A \in  \GO_n(\F_q)$, i.e.\ $AA^\top =\lambda I_n$ for some $\lambda \in \F_q^*$.  Then $GA$ is a generator matrix for $CA$.  Since $GA(GA)^\top= GAA^\top G^\top = G\lambda I_nG^\top=\lambda GG^\top=\lambda {\bf 0}_k= {\bf 0}_k$, we see that $CA$ is self-dual as well, and so (1) holds.

Let $G, \hat{G}$ be generator matrices for $C$ and $\hat{C}$ respectively that are in reduced row echelon form. Since up to column permutation, every code has a systematic generator matrix, we see that $G$ and $\hat{G}$ have the form $[I_k | M] P$ and $[I_k | \hat{M}] \hat{P}$, for some matrices $M, \hat{M} \in \F_q^{k \times k}$ and some permutation matrices $P, \hat{P} \in \O_n(\F_q)$, where $k=\frac{n}{2}$.  Since $C$ and $\hat{C}$ are self-dual, $M M^T= \hat{M}\hat{M}^\top = -I_k$.  Then $GA = \hat{G}$ for 
\[
A = P^{-1} \begin{bmatrix} I_k & {\bf 0}_k \\ {\bf 0}_k & -M^\top \hat{M}\end{bmatrix} \hat{P}.
\]
Since the permutation matrices satisfy $PP^\top=I_n$, we have that $AA^\top=I_n$.  Thus, $A \in \O_n(\F_q)$ and $\hat{C}=CA$, so (2) holds.  
\end{proof}

To make use of this transitive action in our enumeration of self-dual codes, we must begin with some specified canonical self-dual code on which we will act.  
To develop such a canonical self-dual code when $q \equiv 3 \pmod 4$, we will need a straight-forward number theoretic lemma:
\begin{lemma}\cite[ Lemma 11.1, p.\ 138]{Taylor}\label{numbertheory}
If $q=p^e$ is a power of an odd prime, then for any $\alpha \in \F_q$, there exist elements $a, b \in \F_q$ such that $a^2+b^2=\alpha$. 
\end{lemma}

In particular, Lemma \ref{numbertheory} guarantees the existence of elements $a, b \in \F_q$ such that $a^2+b^2=-1$ when $q$ is a power of an odd prime.  It is clear that for a self-dual code to exist in $\F_q^n$, $n$ must be even.  Furthermore, Pless shows in \cite{Pless1968} that for a self-dual code to exist in $\F_q^n$ when $q\equiv3 \pmod 4$, $n$ must be a multiple of 4.  Abiding by these constraints on $n$, we define a \emph{canonical} self-dual code $C_{\Canon}$ in $\F_q^n$ for every value of $n$ that admits a self-dual code.
\begin{definition}\label{Canon}
Let $n$ be even and $k=\frac{n}{2}$.  The  \emph{canonical self-dual code $C_{\Canon}$ in $\F_q^n$} is the code generated by one of the following generator matrices as dictated by the value of $q$ mod 4:
\[
\begin{array}{ll}
\begin{bmatrix} I_k ~\vert~ I_k\end{bmatrix}& \textrm{ if } 2~\vert~q;\\
&\\
\begin{bmatrix} I_k ~\vert~ aI_k\end{bmatrix} &\textrm{ if }q \equiv 1 \pmod 4; \\
&\\
\footnotesize{\left[ \begin{array}{c|c|c|c}I_{\frac{k}{2}} & {\bf 0} & bI_{\frac{k}{2}} & cI_{\frac{k}{2}} \\ {\bf 0} & I_{\frac{k}{2}}& cI_{\frac{k}{2}}& -bI_{\frac{k}{2}} \end{array} \right]}& \textrm{ if }q \equiv 3 \pmod 4,\ \ 4~\vert~n,
\end{array}
\]
where $a \in \F_q$ satisfies $a^2=-1$ when $q \equiv 1 \pmod 4$, and $b,~c \in \F_q$ satisfy $b^2~+~c^2~=~-1$ when $q \equiv 3 \pmod 4$.
\end{definition}
It is easy to check that for each value of $q$ and $n$, the rows of the generator matrices given in Definition \ref{Canon} are pairwise orthogonal and linearly independent.  Note that different values of $a,$ $b$, and $c$ satisfying the prescribed conditions will give distinct, possibly inequivalent, canonical self-dual codes; however, the enumeration method outlined below is independent of these values, and so we may choose any values for these parameters to successfully carry out the enumeration.
We now characterize the set of elements in $\GL_n(\F_q)$ that fix the canonical self-dual code.  We use this characterization to determine which elements of $\GO_n(\F_q)$ fix that code.  This is a key step to determining all the linearly inequivalent self-dual codes.

\begin{definition}\label{LGOStabDef}
Let $C \subseteq \F_q^n$ be a self-dual code.  The \emph{linear stabilizer} $\LStab(C)$ of $C$ is given by
\[
\LStab(C) = \{M \in \GL_n(\F_q) ~|~ C M = C\}.
\]

\noindent The \emph{orthogonal linear stabilizer} $\LGOStab(C)$ of $C$ is given by 
\[
\LGOStab(C) = \LStab(C) \cap \GO_n(\F_q).
\]
\end{definition}
\begin{remark}
For $C=\rho(\cC)$ for some $\cC \in \Matlm$, neither $\LStab(C)$ nor $\LGOStab(C)$ coincides with the traditional automorphism group (one extended row vectors) of the code $\cC$ since that group is restricted to only contain \emph{equivalence maps} that fix the code, i.e.\ elements of the form $\langle T\rangle \rtimes \GO_l(\F_q) \otimes \GO_m(\F_q)$ that fix $\cC$.  However, both $\LStab(C)$ and $\LGOStab(C)$ contain the traditional automorphism group on extended row vectors of the code $\cC$.
\end{remark}
Proposition \ref{Stab} below precisely characterizes the elements of $\LStab(C_{\Canon})$ for the canonical self-dual code $C_{\Canon}$.  This is an extension of Theorem 11 in  \cite{Janusz07}, which characterizes the matrices in $\GL_n(\F_2)$ that stabilize a given binary self-dual code.  The proof is similar to that of \cite{Janusz07}, and so we omit it here.
\begin{proposition}[Extended from \cite{Janusz07}]\label{Stab}
Let $C_{\Canon} \subseteq \F_q^n$ be the canonical self-dual code of dimension $k=\frac{n}{2}$ with generator matrix $[I_k~|~M]$ for appropriate $M$.   Let $S=\begin{bmatrix} I_k & M \\  I_k & {\bf 0}_k \end{bmatrix}$.
Then the linear stabilizer of $C_{\Canon}$ is given by 
\[
\LStab(C_{\Canon}) = \left\{\begin{array}{l|l} S^{-1}{\scriptsize\begin{bmatrix} A&0\\ B&C\end{bmatrix}}S& A,C \in \GL_k(\F_q) \text{ and } B \in \F_q^{k \times k} \end{array}\right\}.
\]
\end{proposition}
We can now precisely enumerate which elements of $\GO_n(\F_q)$ will map $C_{\Canon}$ to linearly block-inequivalent or linearly matrix-inequivalent self-dual codes.  
Theorem \ref{doublecosetcharacterization} shows that the linear maps corresponding to linearly inequivalent codes are the representatives of distinct double cosets within the group $\GO_n(\F_q)$.  
This result is analogous to that found over $\F_2$ in \cite{Janusz07}.  Before proving this theorem, we briefly review the group theory topic of double cosets.

\begin{definition}\cite[p.\ 117]{DummitFoote}
Let $H,K \leq G$ and let $g \in G$.  The \emph{$H$-$K$ double coset of $g$} is given by
\[
HgK= \{ hgk~|~ h \in H,~ k \in K\}.
\]
\end{definition}
\begin{remark} 
As with left and right cosets, the collection of $H$-$K$ double cosets partition the group $G$.  To see this, we recast the $H$-$K$ double coset of $g$ as either a union over $h \in H$ of left cosets of the form $hgK$ or alternatively, as a union over $k \in K$ of right cosets of the form $Hgk$.  One important difference between double cosets and left or right cosets is that double cosets need not all have the same size, i.e. $|HgK|$ need not equal $|H\hat{g}K|$ for $g, \hat{g} \in G$.
\end{remark}

Recall that $\GO_n(\F_q)$ acts transitively on the collection of self-dual codes, and so every self-dual code can be written in the form $C= C_{\Canon} A$ for some $A \in \GO_n(\F_q)$.  Theorem \ref{doublecosetcharacterization} below demonstrates that self-dual codes $C_{\Canon} A$ and $C_{\Canon} B$ are linearly inequivalent precisely when $A$ and $B$ lie in distinct double cosets dictated by the appropriate group of equivalence maps.
\begin{theorem}[Extended from \cite{Janusz07}]\label{doublecosetcharacterization}
Let $n=lm$ for some positive integers $l$ and $m$, and let $\rho:\Matlm \to \F_q^n$ be as in Definition \ref{rhomap}.   Let $\LGOStab(C_{\Canon})$ be the orthogonal linear stabilizer as in Definition \ref{LGOStabDef}; let $\cM_n(\F_q)$ be the scalar monomial matrices as in Definition \ref{scalarmondef}; and let $\LEquiv_{\Vector}^{\SD}(\Matlm)$ be the linear matrix-equivalence maps for self-dual codes as in Definition \ref{matrixequivalenceSD}.
Then for any self-dual codes $C= C_{\Canon}A \subseteq \F_q^n$ and $\widehat{C}=C_{\Canon}B\subseteq \F_q^n$ with $A,B \in \GO_n(\F_q)$, we have 
\begin{enumerate}
\item $C$ is linearly block-equivalent to $\widehat{C}$ if and only if $A$ and $B$ are in the same $\LGOStab(C_{\Canon})$-$\cM_n(\F_q)$ double coset of $\GO_n(\F_q)$. 
\item $\rho^{-1}(C)$ is linearly matrix-equivalent to $\rho^{-1}(\widehat{C})$ if and only if $A$ and $B$ are in the same $\LGOStab(C_{\Canon})$-$\LEquiv_{\Vector}^{\SD}(\Matlm)$ double coset of $\GO_n(\F_q)$.  
\end{enumerate}
\end{theorem}
\begin{proof}
By definition, the self-dual codes $C$ and $\widehat{C}$ are linearly block-equivalent if and only if there exists a map $f \in \LEquiv_{\Block}^{\SD}(\F_q^n)$ such that $C= f(\widehat{C})$ or equivalently if and only if there exists an $M \in \cM_n(\F_q)$ such that $C=\hC M$, i.e.  
\[
C_{\Canon}A= C_{\Canon}BM.
\]
This equality holds if and only if $C_{\Canon}= C_{\Canon}BMA^{-1},$ which occurs precisely when $BMA^{-1}$ fixes $C_{\Canon}$, i.e. when
 \[
BMA^{-1} \in \LGOStab(C_{\Canon}).
\]
There exists an $N \in \LGOStab(C_{\Canon})$ such that $BMA^{-1} = N$ if and only if $B = N A M^{-1}$ 
 if and only if 
 \[
 B \in \LGOStab(C_{\Canon}) A \cM_n(\F_q),
 \]
 i.e., $A$ and $B$ are in the same $\LGOStab(C_{\Canon})$-$\cM_n(\F_q)$ double coset of $\GO_n(\F_q)$.  

By Definition \ref{matrixequivalenceSD}, $\rho^{-1}(C)$ is linearly matrix-equivalent to $\rho^{-1}(\widehat{C})$ if and only if there exists an $f~\in~\LEquiv_{\Vector}^{\SD}(\F_q^{l \times m})$ such that $C=f(\hC)$.  By Proposition \ref{LEquivSDVec}, such an $f$ exists if and only if there exists an $M \in \LEquiv_{\Vector}^{\SD}(\Matlm)$ such that $C= \hC M$, and so replacing $\cM_n(\F_q)$ with $\LEquiv_{\Vector}^{\SD}(\Matlm)$ in the previous argument gives the analogous result for matrix-equivalence.
\end{proof}

%
%

We have reduced the problem of enumerating linearly inequivalent self-dual codes to the problem of determining double-coset representatives.  Janusz notes that the problem of finding double-coset representatives in large groups is also intractable, so it is not clear that this reduction has significantly aided in the classification of inequivalent self-dual matrix codes.  We find, however, that this method is more computationally feasible than brute force techniques, and so we are at least able to employ this method to enumerate the inequivalent self-dual matrix codes for small values of $q$, $l$, and $m$.  The results of this enumeration are given in Section \ref{Enumerateselfdual}.

\section{Enumeration of Linearly Matrix-Inequivalent\\ Self-Dual Matrix Codes}\label{Enumerateselfdual}
In this section, we provide a complete enumeration of the linearly inequivalent self-dual matrix codes for small parameters.  To derive this enumeration, we implemented the method outlined in Section \ref{doublecosetsmethod} in the computer algebra system Magma using a built-in function for enumerating representatives of double cosets of matrix groups.   These representatives give maps from the canonical self-dual code to a set of linearly inequivalent self-dual matrix codes.  In Tables \ref{F2SDCodes}-\ref{F5SDCodes} below, we provide generator matrices for the linearly inequivalent self-dual codes that are the images of the canonical self-dual code under these maps, which by Theorem \ref{doublecosetcharacterization} gives a complete enumeration of the linear equivalence classes of self-dual matrix codes. In addition to the equivalence class representatives, we also provide the rank-distance distribution, or equivalently, since the codes are linear, the rank-weight distribution so that we may evaluate the optimality of these codes. We provide these descriptions for every equivalence class of self-dual matrix codes over $\F_2,\F_3, \F_4,$ and $\F_5$ that was computationally feasible in Magma via the double-coset method.

We first focus on an example of the enumeration of the linearly inequivalent self-dual matrix codes.  Recall from Lemma \ref{SelfDualBlockMatrix} that a self-dual matrix code $C \subseteq \Matlm$ gives rise to a self-dual block code $\rho(C) \subseteq \F_q^{lm}$ via the map $\rho$.  Furthermore, by definition, any generator matrix for the matrix code $C$ is also a generator matrix for the block code $\rho(C)$.  Thus, it is natural to compare the linear-matrix-equivalence classes of self-dual matrix codes to the linear-block-equivalence classes of self-dual block codes of the appropriate parameters.  Here, we consider the example of linearly inequivalent self-dual matrix codes in $\F_3^{2 \times 4}$ and compare these to the linearly inequivalent self-dual block codes in $\F_3^8$.  In this context, we explicitly compare the matrix-equivalence classes to the block-equivalence classes in order to gain further insight into the relationship between block- and matrix-equivalence.  This example highlights some key differences in these notions of equivalence, which will prove valuable for comparing the enumerations in the block- and matrix-code case more generally.  

\begin{example}
Set $q=3$, $l=2$, $m=4$.  Following the method outlined above, Magma produces 13 linearly inequivalent self-dual matrix codes in $\F_3^{2 \times 4}$; the generator matrices for these codes are given in Table \ref{F3SDCodes}.  In contrast, Mallows, Pless and Sloane show that there is only one linear equivalence class of self-dual block codes in $\F_3^8$ \cite{MallowsPlessSloane1976}.  Thus, the 13 different self-dual block codes that arise from the generator matrices in Table \ref{F3SDCodes} are all linearly block-equivalent.  To gain some intuition as to why this is true, we focus on the first two linearly-inequivalent matrix codes in Table \ref{F3SDCodes} and examine them both as matrix codes and as block codes. 

Let $C_1 \subseteq \F_3^{2 \times 4}$ be the matrix code generated by $G_1$ in Table \ref{F3SDCodes}.  Similarly, let $C_2 \subseteq \F_3^{2 \times 4}$ be the matrix code generated by $G_2$ in Table \ref{F3SDCodes}.  Since these matrix codes have different rank-weight distributions, there cannot be a rank-weight-preserving map from one code to the other, and so these codes cannot be matrix equivalent.  

We can check that both $\rho(C_1)$ and $\rho(C_2)$ have Hamming-weight distributions with $A_0=1$, $A_3=16$, $A_6=64$ and the remaining $A_i$ equal to zero, where $A_i$ denotes the number of codewords of weight $i$.  Thus, it is possible that there exists a Hamming-weight-preserving map between the codes.  In fact, one can check via Magma, for example, that the matrix $M$ given by
\[
M={\scriptsize \begin{bmatrix}  1&0&0&0&0&0&0&0\\
 0&0&1&0&0&0&0&0\\
 0&0&0&0&0&0&2&0\\
 0&2&0&0&0&0&0&0\\
 0&0&0&0&2&0&0&0\\
 0&0&0&0&0&0&0&1\\
 0&0&0&0&0&2&0&0\\
 0&0&0&2&0&0&0&0
\end{bmatrix}} \in \GL_8(\F_3)
\]
maps $\rho(C_1)$ to $\rho(C_2)$, and since $M$ is a monomial matrix of 1s and 2s, i.e. 1s and $-1$s over $\F_3$, $M$ is a linear block-equivalence map.  

Intuitively, $\rho(C_1)$ and $\rho(C_2)$ are linearly block-equivalent because it is possible to permute and scale the entries of one code to obtain the other, and the actions of permuting and scaling entries do not affect the Hamming weight of a vector; however, these same actions will affect the rank weight of the matrix obtained from the extended row vector, and so the corresponding matrix codes $C_1$ and $C_2$ are not equivalent.  

In general, when we restrict to linear equivalence maps for self-dual codes, there are significantly more block-equivalence maps for self-dual codes than there are matrix-equivalence maps for self-dual codes.  In this example, applying a simple counting argument\footnote{Formulas derived from such a counting argument are available in \cite{MyThesis} for the interested reader.} and a formula for the size of $O_n(\F_q)$ \cite{MacWilliamsOrthogonal} to the groups of equivalence maps described in Definition \ref{scalarmondef} and Proposition \ref{LEquivSDVec}, we see that 
\[
|\LEquiv_{\Block}^{\SD}(\F_3^8)| = 10,321,920 > 18,432 = |\LEquiv_{\Mat}^{\SD}(\F_3^{2 \times 4})|
\]
demonstrating that linear block-equivalence for self-dual codes is much broader than linear matrix-equivalence for self-dual codes.  Thus, we would expect many more equivalence classes for self-dual matrix codes than for self-dual block codes.  

It is worth noting that while there are significantly more linear block-equivalence maps for self-dual codes than there are linear matrix-equivalence maps for self-dual codes, this phenomenon does not hold true for general linear block-equivalence and linear matrix-equivalence maps.  If we do not restrict to those maps that commute with the dual, we see that we do not gain any linear block-equivalence maps because all monomial matrices over $\F_3$ consist of only $\pm 1$, but we have drastically more linear matrix-equivalence maps:
\[
|\LEquiv_{\Block}(\F_3^8)| = 10,321,920 < 582,266,880 = |\LEquiv_{\Mat}(\F_3^{2 \times 4})|.
\]
Thus, in the general case, we would expect to see far fewer linear matrix-equivalence classes than linear block-equivalence classes, even though this behavior is reversed in the case of self-dual codes.
\end{example}

Another important question to consider is whether any self-dual codes are optimal with respect to their minimum distance.  The analogue of the Singleton bound for matrix codes \cite{Delsarte} guarantees that for any $l \times m$ matrix code, we must have
\[
k \leq \min\{l, m\}(\max\{l, m\} - d +1)
\]
where $k$ is the dimension of the code and $d$ is its minimum distance. Using that $l \leq m$ and $k =\dfrac{l m}{2}$, this simplifies to $d \leq \frac{l}{2} +1$.  In the enumerations that follow, $l=2$, and so the minimum distance of any \emph{maximum rank distance} (MRD) code, i.e.\ any code that meets the analogue of the Singleton bound, with these parameters is 2.  As the tables below indicate, there are a number of parameters for which there is a self-dual code that is MRD; however, for some parameters, no self-dual codes achieve this upper bound.  In the latter case, the self-dual codes are all suboptimal, and there exist non-self dual codes with the same parameters that are MRD.\footnote{MRD codes can be constructed for any parameters using a modification of a construction from \cite{Blaum}.}  It is not yet clear why/when certain code parameters will not yield any self-dual codes that are MRD, but this is an interesting question for future research.

\section{Conclusion and Future Work}
In this paper, we developed a method to enumerate all inequivalent self-dual matrix codes of relatively short length over small finite fields.
This method is only computationally feasible for particularly small parameters because generating a set of representatives of double cosets is also computationally intractable for large parameters.  An open question in this area is whether this method may be modified to produce an enumeration of self-dual codes over larger fields and/or larger dimensional vector spaces.  One possible modification appeals to the existing enumeration of block-inequivalent self-dual block codes.  It is possible that we may exploit the relationship between the block-equivalence and matrix-equivalence maps to generate representatives of the equivalence classes determined by the intersection of the groups of block-equivalence and matrix-equivalence maps using the enumeration of block-inequivalent self-dual codes as a starting point.  This set of representatives would contain all the representatives of the matrix-equivalence classes as well as some potentially extraneous codes, but it should be computationally feasible to then weed through that list of representatives to find the subset of matrix-inequivalent codes.  


\begin{table}[ht]
\caption{Enumeration of Linearly Matrix-Inequivalent Self-Dual Matrix Codes over $\F_2$}
\centering
{\footnotesize
\begin{tabular}{ccccccl}
\hline
\hline
\vspace{-.1in}
\\
\vspace{.03in} 
$q$&$l$&$m$&$n$&
{\footnotesize \begin{tabular}{c} Number of \\Equiv.\ Classes\end{tabular}}&  
{\footnotesize \begin{tabular}{c} Equivalence Class\\Representatives \end{tabular}} \phantom{h}&\hspace{.03in} {\footnotesize \begin{tabular}{c} Weight \\ Distributions\end{tabular}}
\\
\hline 
\vspace{-.08in}\\
\vspace{.05in}
2&2&2&4&2&\hspace{-.15in}{\scriptsize $ G_{1}=\begin{bmatrix} 1&0&1&0\\
 0&1&0&1
\end{bmatrix}$}
&{\small $ A_0=1, A_{1}=3$}\\

\vspace{.1in}&&&&&\hspace{-.15in}{\scriptsize $ G_{2}=\begin{bmatrix} 1&0&0&1\\
 0&1&1&0
\end{bmatrix}$}
&{\small $ A_0=1, A_{1}=1, A_{2}=2$}
\\
\hline 
\vspace{-.08in}\\

2&2&3&6&5& 
\vspace{-.16in}\\
\vspace{.05in}&&&&&\hspace{-.15in}{\scriptsize $ G_{1}=\begin{bmatrix} 1&0&0&1&0&0\\
 0&1&0&0&1&0\\
 0&0&1&0&0&1
\end{bmatrix}$} 
&{\small $ A_0=1, A_{1}=7$} \\

\vspace{.05in}&&&&&\hspace{-.15in}{\scriptsize $ G_{2}=\begin{bmatrix} 0&0&1&1&0&0\\
 1&1&0&0&0&0\\
 0&0&0&0&1&1
\end{bmatrix}$} 
&{\small $ A_0=1, A_{1}=3, A_{2}=4$} \\

\vspace{.05in}&&&&&\hspace{-.15in}{\scriptsize $ G_{3}=\begin{bmatrix} 1&0&0&0&0&1\\
 0&1&0&1&0&0\\
 0&0&1&0&1&0
\end{bmatrix}$} 
&{\small $ A_0=1, A_{1}=1, A_{2}=6$} \\

\vspace{.05in}&&&&&\hspace{-.15in}{\scriptsize $ G_{4}=\begin{bmatrix} 1&0&0&1&0&0\\
 0&1&1&0&0&0\\
 0&0&0&0&1&1
\end{bmatrix}$} 
&{\small $ A_0=1, A_{1}=5, A_{2}=2$} \\

\vspace{.1in}&&&&&\hspace{-.15in}{\scriptsize $ G_{5}=\begin{bmatrix} 1&0&0&0&0&1\\
 0&1&0&0&1&0\\
 0&0&1&1&0&0
\end{bmatrix}$} 
&{\small $ A_0=1, A_{1}=3, A_{2}=4$}\\

\\
\hline
\hline
\end{tabular}}
\label{F2SDCodes}
\end{table}

\begin{table}[ht]
\vspace{-.1in}
\caption{Enumeration of Linearly Matrix-Inequivalent Self-Dual Matrix Codes over $\F_3$}
\centering
{\footnotesize
\begin{tabular}{cccccrl}
\hline
\hline \vspace{-.1in}
\\
\vspace{.03in} 
$q$&$l$&$m$&$n$&
{\footnotesize \begin{tabular}{c} Number of \\Equiv.\ Classes\end{tabular}}&  
{\footnotesize \begin{tabular}{c} Equivalence Class\\Representatives \end{tabular}} \phantom{hell}&\hspace{.03in} {\footnotesize \begin{tabular}{c} Weight \\ Distributions\end{tabular}}
\\
\hline 
\vspace{-.08in}\\
\vspace{.1in}
 3&2&2&4& 1& {\scriptsize $ G_1=\begin{bmatrix}1&0&2&2\\
 0&1&2&1\end{bmatrix}$}\phantom{hello!!}
 & {\small $A_0=1,~A_2=8$}\\
 
\hline 
\vspace{-.06in}\\

3&2&4&8&13& 
\vspace{-.2in}\\
&&&&&\hspace{-.28in}{\scriptsize $ G_1=\begin{bmatrix}1&0&0&0&2&0&2&0\\
 0&1&0&0&0&2&0&2\\
 0&0&1&0&2&0&1&0\\
 0&0&0&1&0&2&0&1\end{bmatrix}$}
& {\small $A_0=1,~ A_2=80$} \\

\vspace{.05in} &&&&&\hspace{-.3in} {\scriptsize $ G_2=\begin{bmatrix}1&0&1&1&2&0&2&2\\
 0&1&0&1&1&2&2&1\\
 0&0&1&1&0&0&0&2\\
 0&0&0&0&1&2&2&0\end{bmatrix}$}
& {\small $A_0=1,~ A_1=8,~A_2=72$} \\

\vspace{.05in} &&&&&\hspace{-.3in} {\scriptsize $ G_3=\begin{bmatrix}1&0&1&1&1&0&2&2\\
 0&2&2&1&1&1&1&0\\
 0&0&0&2&0&0&1&1\\
 0&0&2&1&1&1&2&2\end{bmatrix}$}
& {\small $A_0=1,~ A_2=80$} \\

\vspace{.05in} &&&&&\hspace{-.3in} {\scriptsize $ G_4=\begin{bmatrix}1&0&0&0&0&0&1&1\\
 0&1&0&2&2&1&2&1\\
 0&0&0&2&0&0&2&1\\
 0&0&1&0&1&1&0&0\end{bmatrix}$}
&{\small $A_0=1,~A_1=4,~A_2=76$}\\

\vspace{.05in} &&&&&\hspace{-.3in} {\scriptsize $ G_5=\begin{bmatrix}1&0&0&1&1&2&1&2\\
 0&2&2&0&2&2&1&1\\
 0&0&0&1&0&0&1&2\\
 0&0&2&0&0&0&1&1\end{bmatrix}$}
& {\small $A_0=1,~ A_2=80$} \\

\vspace{.05in} &&&&&\hspace{-.3in}{\scriptsize $ G_6=\begin{bmatrix}1&0&0&2&1&0&0&0\\
 0&2&1&0&0&2&0&0\\
 0&0&0&1&1&0&0&2\\
 0&0&2&0&0&2&1&0\end{bmatrix}$}
& {\small $A_0=1,~ A_2=80$} \\

\vspace{.05in} &&&&&\hspace{-.3in} {\scriptsize $ G_7=\begin{bmatrix}1&0&1&1&2&0&2&2\\
 0&2&1&2&0&1&2&1\\
 0&0&0&1&2&0&2&0\\
 0&0&1&0&0&1&0&1\end{bmatrix}$}
&{\small $A_0=1,~A_1=16,~A_2=64$}\\

\vspace{.05in} &&&&&\hspace{-.3in} {\scriptsize $ G_8=\begin{bmatrix}1&0&1&1&1&0&2&2\\
 0&1&2&1&0&1&2&1\\
 0&0&1&1&1&0&0&0\\
 0&0&0&0&0&1&2&1\end{bmatrix}$}
&{\small $A_0=1,~A_1=8,~A_2=72$}\\

&&&&&\hspace{-.28in}{\scriptsize $ G_9=\begin{bmatrix}1&0&2&2&2&0&1&1\\
 0&2&2&1&0&2&2&1\\
 0&0&0&0&2&0&1&1\\
 0&0&0&0&0&2&2&1\end{bmatrix}$}
&{\small $A_0=1,~A_1=32,~A_2=48$}\\

\vspace{.05in} &&&&&\hspace{-.3in}{\scriptsize $ G_{10}=\begin{bmatrix}1&0&0&0&2&0&2&0\\
 0&1&0&0&0&2&0&2\\
 0&0&2&0&2&0&1&0\\
 0&0&0&1&0&2&0&1\end{bmatrix}$}
&{\small $A_0=1,~A_2=80$}\\

\vspace{.05in} &&&&&\hspace{-.3in} {\scriptsize $ G_{11}=\begin{bmatrix}1&0&2&1&2&0&2&2\\
 0&1&0&1&1&2&2&1\\
 0&0&2&1&0&0&0&2\\
 0&0&0&0&1&2&2&0\end{bmatrix}$}
&{\small $A_0=1,~A_1=8,~A_2=72$}\\

\vspace{.05in} &&&&&\hspace{-.3in}{\scriptsize $ G_{12}=\begin{bmatrix}1&0&0&0&0&0&1&1\\
 0&1&0&2&2&1&2&1\\
 0&0&0&2&0&0&2&1\\
 0&0&2&0&1&1&0&0\end{bmatrix}$}
&{\small $A_0=1,~A_1=4,~A_2=76$}\\

\vspace{.1in} &&&&&\hspace{-.3in}{\scriptsize $ G_{13}=\begin{bmatrix}1&0&1&1&1&0&1&1\\
 0&1&1&2&0&1&2&1\\
 0&0&0&0&1&0&1&1\\
 0&0&0&0&0&1&2&1\end{bmatrix}$}
&{\small $A_0=1,~A_1=20,~A_2=60$}
\\
%
\hline
\hline 
\end{tabular}}
\label{F3SDCodes}
\end{table}

\begin{table}[ht]
\caption{Enumeration of Linearly Matrix-Inequivalent Self-Dual Matrix Codes over $\F_4~=~\F_2[\alpha]$ where $\alpha^2+\alpha+1=0$}
\centering
{\footnotesize
\begin{tabular}{cccccrl}
\hline
\hline \vspace{-.1in}
\\
\vspace{.03in} 
$q$&$l$&$m$&$n$&
{\footnotesize \begin{tabular}{c} Number of \\Equiv.\ Classes\end{tabular}}&  
{\footnotesize \begin{tabular}{c} Equivalence Class\\Representatives \end{tabular}} &\hspace{.03in} {\footnotesize \begin{tabular}{c} Weight \\ Distributions\end{tabular}}
\\
\hline \vspace{-.08in}\\
\vspace{-.15in}
4&2&2&4&3\\
\vspace{.05in}&&&&&{\scriptsize $ G_{1}=\begin{bmatrix}\phantom{l}1\phantom{l}&\phantom{l}0\phantom{l}&\phantom{l}1\phantom{l}&\phantom{l}0\phantom{l}\\
 0&1&0&1\end{bmatrix}$}
&{\small $ A_{0}=1, A_{1}=15$}\\

\vspace{.05in}&&&&&{\scriptsize $ G_{2}=\begin{bmatrix}\phantom{l}1\phantom{l}&\phantom{l}0\phantom{l}&\alpha&\alpha^2\\
 0&1&\alpha^2&\alpha\end{bmatrix}$}
&{\small $ A_{0}=1, A_{1}=3, A_{2}=12$}\\

\vspace{.1in}&&&&&{\scriptsize $ G_{3}=\begin{bmatrix}\phantom{l}1\phantom{l}&\phantom{l}0\phantom{l}&\phantom{l}0\phantom{l}&\phantom{l}1\phantom{l}\\
 0&1&1&0\end{bmatrix}$}
&{\small $ A_{0}=1, A_{1}=3, A_{2}=12$}\\

\hline
\hline
\end{tabular}}
\label{F4SDCodesLinear}
\end{table}

\begin{table}[ht]
\caption{Enumeration of Linearly Matrix-Inequivalent Self-Dual Matrix Codes over $\F_5$}
\centering
{\footnotesize
\begin{tabular}{cccccrl}
\hline
\hline \vspace{-.1in}
\\
\vspace{.03in} 
$q$&$l$&$m$&$n$&
{\footnotesize \begin{tabular}{c} Number of \\Equiv.\ Classes\end{tabular}}&  
{\footnotesize \begin{tabular}{c} Equivalence Class\\Representatives \end{tabular}} \phantom{hel}&\hspace{.03in} {\footnotesize \begin{tabular}{c} Weight \\ Distributions\end{tabular}}
\\
\hline 
\vspace{-.08in}\\
\vspace{.05in}
5&2&2&4&2&{\scriptsize $ G_1=\begin{bmatrix}1&0&3&0\\
 0&1&0&3\end{bmatrix}$} \phantom{hello!!}&
{\small $A_0=1,~A_1=24$}\\
\vspace{.1in} &&&&& {\scriptsize $ G_2=\begin{bmatrix}0&0&4&2\\
 2&4&2&1\end{bmatrix}$} \phantom{hello!!}&
{\small $A_0=1,~A_1=8,~A_2=16$}
\\

\hline \vspace{-.05in}\\

\vspace{-.17in}
5&2&3&6&7&\\
\vspace{.05in} &&&&& \hspace{-.5in}{\scriptsize $ G_1=\begin{bmatrix}1&0&0&3&0&0\\
 0&1&0&0&3&0\\
 0&0&1&0&0&3\end{bmatrix}$}\phantom{hel}
&{\small $A_0=1,~A_1=124$}
\\
\vspace{.05in} &&&&& \hspace{-.5in}{\scriptsize $ G_2=\begin{bmatrix}1&0&3&2&1&0\\
 0&1&4&1&1&1\\
 0&0&3&1&4&3\end{bmatrix}$}\phantom{hel}
&{\small $A_0=1,~A_1=4,~A_2=120$}
\\
\vspace{.05in} &&&&& \hspace{-.5in}{\scriptsize $ G_3=\begin{bmatrix}1&0&3&0&1&2\\
 0&1&1&2&2&0\\
 0&0&2&3&1&4\end{bmatrix}$}\phantom{hel}
&{\small $A_0=1,~A_1=4,~A_2=120$}
\\
\vspace{.05in} &&&&& \hspace{-.5in}{\scriptsize $ G_4=\begin{bmatrix}1&0&2&4&2&0\\
 0&4&4&1&4&1\\
 0&0&3&3&1&1\end{bmatrix}$}\phantom{hel}
&{\small $A_0=1,~A_1=28,~A_2=96$}
\\
\vspace{.05in} &&&&& \hspace{-.5in}{\scriptsize $ G_5=\begin{bmatrix}1&0&0&3&0&0\\
 0&1&2&0&0&0\\
 0&0&0&0&1&3\end{bmatrix}$}\phantom{hel}
&{\small $A_0=1,~A_1=12,~A_2=112$}
\\
\vspace{.05in} &&&&& \hspace{-.5in}{\scriptsize $ G_6=\begin{bmatrix}1&0&3&2&1&0\\
 0&4&4&1&1&1\\
 0&0&3&1&4&3\end{bmatrix}$}\phantom{hel}
&{\small $A_0=1,~A_1=28,~A_2=96$}
\\
\vspace{.1in} &&&&& \hspace{-.5in}{\scriptsize $ G_7=\begin{bmatrix}1&0&3&0&0&0\\
 0&4&0&0&3&0\\
 0&0&0&1&0&3\end{bmatrix}$}\phantom{hel}
&{\small $A_0=1,~A_1=44,~A_2=80$}
\\


\hline
\hline
\end{tabular}}

\label{F5SDCodes1}
\end{table}

\begin{table}[ht]
\caption{Enumeration of Linearly Matrix-Inequivalent Self-Dual Matrix Codes over $\F_5$ (Continued)}
\centering
{\footnotesize \begin{tabular}{cccccrl}
\hline
\hline \vspace{-.1in}
\\
\vspace{.03in} 
$q$&$l$&$m$&$n$&
{\footnotesize \begin{tabular}{c} Number of \\Equiv.\ Classes\end{tabular}}&  
{\footnotesize \begin{tabular}{c} Equivalence Class\\Representatives \end{tabular}} \phantom{hel}&\hspace{.03in} {\footnotesize \begin{tabular}{c} Weight \\ Distributions\end{tabular}}
\\
\hline \vspace{-.08in}\\
\vspace{-.18in}
5&2&4&8&24&\\
\vspace{.05in}&&&&&\hspace{-.5in}{\scriptsize $ G_{1}=\begin{bmatrix}1&0&0&0&3&0&0&0\\
 0&1&0&0&0&3&0&0\\
 0&0&1&0&0&0&3&0\\
0&0&0&1&0&0&0&3\end{bmatrix}$}  
&{\small $A_0=1,~A_1=624$}\\

\vspace{.05in}&&&&&\hspace{-.5in}{\scriptsize $ G_{2}=\begin{bmatrix}  1&2&3&4&1&1&3&2\\
 1&1&4&4&3&0&1&4\\
 1&4&0&1&4&1&1&2\\
1&0&4&0&3&2&1&2\end{bmatrix}$}  
&{\small $A_0=1,~A_1=48,~A_2=576$}\\

\vspace{.05in}&&&&&\hspace{-1in}{\scriptsize $ G_{3}=\begin{bmatrix}  2&4&1&3&3&2&1&4\\
 4&3&2&3&3&1&4&1\\
 4&4&0&1&4&1&1&2\\
2&3&1&1&0&1&3&0\end{bmatrix}$}  
&{\small $A_0=1,~A_1=48,~A_2=576$}\\

\vspace{.05in}&&&&&\hspace{-.5in}{\scriptsize $ G_{4}=\begin{bmatrix}  2&0&0&0&0&3&4&1\\
 3&1&4&4&3&3&0&0\\
 1&0&4&3&0&3&2&1\\
0&0&4&0&2&1&3&0\end{bmatrix}$} 
&{\small $A_0=1,~A_1=24,~A_2=600$}\\

\vspace{.05in}&&&&&\hspace{-.5in}{\scriptsize $ G_{5}=\begin{bmatrix}  2&4&4&1&2&1&3&2\\
 1&3&1&2&1&3&0&0\\
 1&3&2&0&0&1&1&2\\
1&2&1&3&1&0&0&3\end{bmatrix}$} 
&{\small $A_0=1,~A_1=32,~A_2=592$}\\

\vspace{.05in}&&&&&\hspace{-.5in}{\scriptsize $ G_{6}=\begin{bmatrix}  0&1&1&4&1&2&1&4\\
 1&1&0&0&3&3&0&0\\
 2&3&3&1&4&1&1&2\\
 3&3&2&2&2&1&3&0\end{bmatrix}$} 
&{\small $A_0=1,~A_1=64,~A_2=560$}\\

\vspace{.05in}&&&&&\hspace{-.5in}{\scriptsize $ G_{7}=\begin{bmatrix}  3&4&1&3&1&1&3&2\\
 3&3&2&3&0&3&0&0\\
 4&1&3&0&0&3&2&1\\
 0&4&0&3&2&3&4&4\end{bmatrix}$} 
&{\small $A_0=1,~A_1=24,~A_2=600$}\\

\vspace{.05in}&&&&&\hspace{-.5in}{\scriptsize $ G_{8}=\begin{bmatrix}  4&4&0&2&0&1&3&2\\
 4&1&0&0&0&1&4&1\\
 2&0&1&0&2&2&4&4\\
 3&0&3&4&1&1&3&0\end{bmatrix}$} 
&{\small $A_0=1,~A_2=624$}\\

\vspace{.05in}&&&&&\hspace{-.5in}{\scriptsize $ G_{9}=\begin{bmatrix}  0&0&3&3&1&2&1&4\\
 3&2&3&1&3&1&4&1\\
 0&1&4&1&1&2&4&4\\
 3&3&1&1&1&2&1&2\end{bmatrix}$} 
&{\small $A_0=1,~A_1=16,~A_2=608$}\\

\vspace{.05in}&&&&&\hspace{-.5in}{\scriptsize $ G_{10}=\begin{bmatrix}  4&1&3&1&2&4&2&3\\
 3&4&4&3&1&3&0&0\\
 0&3&0&3&1&1&1&2\\
 3&4&2&0&4&1&3&0\end{bmatrix}$} 
&{\small $A_0=1,~A_1=24,~A_2=600$}\\

\vspace{.05in}&&&&&\hspace{-.5in}{\scriptsize $ G_{11}=\begin{bmatrix}  0&2&4&0&2&2&1&4\\
 1&0&3&0&1&3&0&0\\
 4&3&0&3&0&2&4&4\\
 2&2&3&0&3&2&1&2\end{bmatrix}$} 
&{\small $A_0=1,~A_2=624$}\\

\vspace{.1in}&&&&&\hspace{-1in}{\scriptsize $ G_{12}=\begin{bmatrix}  3&4&2&4&1&1&3&2\\
 3&0&4&1&4&1&4&1\\
 2&4&2&3&4&1&1&2\\
 3&3&2&2&2&1&3&0\end{bmatrix}$} 
&{\small $A_0=1,~A_1=8,~A_2=616$}\\

\\\hline
\hline
\end{tabular}}
\label{F5SDCodes2}
\end{table}

\begin{table}[ht]
\caption{Enumeration of Linearly Matrix-Inequivalent Self-Dual Matrix Codes over $\F_5$ (Continued)}
\centering
{\footnotesize \begin{tabular}{cccccrl}
\hline
\hline \vspace{-.1in}
\\
\vspace{.03in} 
$q$&$l$&$m$&$n$&
{\footnotesize \begin{tabular}{c} Number of \\Equiv.\ Classes\end{tabular}}&  
{\footnotesize \begin{tabular}{c} Equivalence Class\\Representatives \end{tabular}} \phantom{hel}&\hspace{.03in} {\footnotesize \begin{tabular}{c} Weight \\ Distributions\end{tabular}}
\\
\hline \vspace{-.08in}\\
\vspace{-.18in}
5&2&4&8&continued&\\

\vspace{.05in}&&&&&\hspace{-.5in}{\scriptsize $ G_{13}=\begin{bmatrix}  4&1&1&4&4&0&0&0\\
 2&1&1&1&4&2&2&3\\
 3&3&2&0&2&0&3&0\\
 2&4&3&1&4&2&1&2\end{bmatrix}$} 
&{\small $A_0=1,~A_2=624$}\\

\vspace{.05in}&&&&&\hspace{-.5in}{\scriptsize $ G_{14}=\begin{bmatrix}  0&1&0&3&4&1&3&2\\
 3&1&1&1&2&3&0&0\\
 2&0&1&0&4&0&3&0\\
 0&3&3&3&3&0&0&3\end{bmatrix}$} 
&{\small $A_0=1,~A_2=624$}\\

\vspace{.05in}&&&&&\hspace{-.5in}{\scriptsize $ G_{15}=\begin{bmatrix}  0&3&1&0&0&0&0&0\\
 4&0&0&2&4&4&3&2\\
 3&3&1&4&3&2&4&4\\
 2&2&4&1&1&0&0&3\end{bmatrix}$} 
&{\small $A_0=1,~A_1=144,~A_2=480$}\\

\vspace{.05in}&&&&&\hspace{-.5in}{\scriptsize $ G_{16}=\begin{bmatrix}  2&4&1&3&2&2&1&4\\
 3&4&2&1&1&3&0&0\\
 4&4&2&3&2&1&1&2\\
 3&4&0&3&0&4&2&1\end{bmatrix}$} 
&{\small $A_0=1,~A_2=624$}\\

\vspace{.05in}&&&&&\hspace{-.5in}{\scriptsize $ G_{17}=\begin{bmatrix}  4&4&2&4&3&1&3&2\\
 0&0&1&3&1&3&0&0\\
 4&3&2&0&0&1&1&2\\
 0&4&2&0&0&1&3&0\end{bmatrix}$} 
&{\small $A_0=1,~A_1=48,~A_2=576$}\\

\vspace{.05in}&&&&&\hspace{-.5in}{\scriptsize $ G_{18}=\begin{bmatrix}  1&0&0&0&3&0&0&0\\
 2&4&0&1&1&0&3&3\\
 2&3&1&1&1&2&1&3\\
 0&0&0&1&0&0&0&3\end{bmatrix}$} 
&{\small $A_0=1,~A_1=128,~A_2=496$}\\

\vspace{.05in}&&&&&\hspace{-.5in}{\scriptsize $ G_{19}=\begin{bmatrix}  1&2&3&4&1&1&3&2\\
 2&0&4&2&1&1&0&3\\
 1&4&0&1&4&1&1&2\\
 0&1&4&2&0&1&2&3\end{bmatrix}$} 
&{\small $A_0=1,~A_1=8,~A_2=616$}\\

\vspace{.05in}&&&&&\hspace{-.5in}{\scriptsize $ G_{20}=\begin{bmatrix}  0&1&1&2&2&0&3&1\\
 4&3&2&3&3&1&4&1\\
 2&1&0&0&3&4&3&4\\
 2&3&1&1&0&1&3&0\end{bmatrix}$} 
&{\small $A_0=1,~A_1=8,~A_2=616$}\\

\vspace{.05in}&&&&&\hspace{-1in}{\scriptsize $ G_{21}=\begin{bmatrix}  2&0&4&2&3&4&0&4\\
 0&0&1&3&1&3&0&0\\
 4&0&2&2&1&2&0&1\\
 2&1&0&0&4&2&2&4\end{bmatrix}$} 
&{\small $A_0=1,~A_1=16,~A_2=608$}\\

\vspace{.05in}&&&&&\hspace{-.5in}{\scriptsize $ G_{22}=\begin{bmatrix}  1&4&2&0&4&3&0&3\\
 4&1&4&3&2&0&0&2\\
 3&4&0&0&2&1&0&0\\
 2&3&3&2&4&4&1&4\end{bmatrix}$} 
&{\small $A_0=1,~A_1=24,~A_2=600$}\\

\vspace{.05in}&&&&&\hspace{-.5in}{\scriptsize $ G_{23}=\begin{bmatrix}  0&4&3&0&2&3&4&1\\
 2&0&0&1&2&2&1&1\\
 2&4&3&1&4&4&2&2\\
 2&4&3&1&3&3&4&4\end{bmatrix}$} 
&{\small $A_0=1,~A_1=64,~A_2=560$}\\

\vspace{.05in}&&&&&\hspace{-.5in}{\scriptsize $ G_{24}=\begin{bmatrix}  0&4&2&3&1&0&4&3\\
 4&3&4&0&1&1&4&1\\
 4&0&4&1&3&3&0&2\\
 1&0&0&1&3&0&0&3\end{bmatrix}$} 
&{\small $A_0=1,~A_1=144,~A_2=480$}\\
\\\hline
\hline
\end{tabular}}
\label{F5SDCodes}
\end{table}

\FloatBarrier

\appendix
\section*{Appendix: Details of Implementation of $\GO_n(\F_q)$ in Magma}
\label{GOncomputation}
\subsection*{Construction of $\GO_n(\F_q)$ for odd $q$}
Magma implements the classical group theory definition of the orthogonal group, which preserves a bilinear form that is not the dot product; the group that preserves that bilinear form modulo scalars is referred to as the \emph{conformal orthogonal group} $\CO_n(\F_q)$.  When $n$ is odd, there is precisely one conformal orthogonal group for the specified bilinear form, and this group is isomorphic to $\GO_n(\F_q)$ as we define it.  When $n$ is even, there are two conformal orthogonal groups $\CO_n^+(\F_q)$ and $\CO_n^-(\F_q)$, where the plus and minus refer to the Witt defect of the underlying bilinear form \cite[p. 136-141]{Taylor}.  The Witt defect of the dot product determines which of these groups is isomorphic to $\GO_n(\F_q)$.  Following the background given in \cite{Taylor}, we see that 
\begin{itemize}
\item if $q \equiv 1 \pmod 4$, then $\GO_{n}(\F_q)$ is isomorphic to $\CO_n^+(\F_q)$,
\item if $q \equiv 3 \pmod 4$ and $n \equiv 0 \pmod 4$, then $\GO_{2m}(\F_q)$ is isomorphic to $\CO_{2m}^+(\F_q)$, and
\item if $q \equiv 3 \pmod 4$ and $n \equiv 2 \pmod 4$, then $\GO_{2m}(\F_q)$ is isomorphic to $\CO_{2m}^-(\F_q)$.
\end{itemize}
\subsection*{Construction of $\GO_n(\F_{q})$ for even $q$}
Over characteristic 2, the dot product is a singular bilinear form whereas the bilinear form defining $\CO_n(\F_q)$ is nonsingular.  Thus, these groups cannot be isomorphic, and so we resort to a random construction of $\GO_{n}(\F_q)$ in this case.  Specifically, we generate a random element  $A \in \GL_n(\F_{q})$ and test whether it satisfies $A A^\top = \lambda I_n$ for some $\lambda \in \F_q^*$.  If $A$ satisfies this condition then we include it in the list of generators for $\GO_{n}(\F_q)$. We continue this process until the cardinality of the subgroup generated by this list of generators equals the cardinality of $\GO_n(\F_{q})$, which can be obtained using formulas for the size of $\O_n(\F_{2^e})$ from \cite{MacWilliamsOrthogonal} and the fact that $\GO_n(\F_{2^e})$ is the normalizer of $\O_n(\F_{2^e})$ in $\GL_n(\F_{2^e})$\footnote{A precise formula for the size of $\GO_n(\F_{2^e})$ is available in \cite{MyThesis}.}.

\section*{Acknowledgements}
The author would like to express her sincere gratitude to the reviewers for their significant advice on ways to streamline the paper and dramatically simplify some proofs.  She would also like to acknowledge that this work stemmed from her Ph.D. thesis, which was completed under the supervision of Judy L.\ Walker in the Department of Mathematics at  University of Nebraska-Lincoln.  Additionally, the author was supported by the following sources during the completion of this work: NSF grants DMS-0735099, DMS-0903517 and DMS-0838463, as well as DOE grant P200A090002. 

\bibliographystyle{plain}
 \bibliography{Thesis_Bibliography}

\end{document}